%% file: main.tex
\documentclass[conference]{IEEEtran}
\usepackage{times}
\usepackage{color, colortbl, xcolor}
\input{notation}
\usepackage[numbers,sort&compress]{natbib}
\usepackage{multicol}
\usepackage[bookmarks=true]{hyperref}
\usepackage{amsfonts}
\usepackage{amsmath}
\usepackage{dsfont}
\usepackage{subfigure}
\usepackage{multicol}
\usepackage[pdftex]{graphicx}   
\usepackage{graphicx}

\usepackage{geometry}
\newgeometry{top=1in,bottom=.75in,right=0.75in,left=.75in}

\usepackage{amsmath}

\graphicspath{{./figs/}}    
\DeclareGraphicsExtensions{.pdf,.png}  
\input{macros.tex}

\usepackage[textfont=md,font=footnotesize]{caption}

\usepackage{ifthen}
\newboolean{include-notes}
\setboolean{include-notes}{true}

\usepackage{lipsum}

\newcommand\blfootnote[1]{%
  \begingroup
  \renewcommand\thefootnote{}\footnote{#1}%
  \addtocounter{footnote}{-1}%
  \endgroup
}

\newcommand{\sgnote}[1]%
    {\textcolor{green}{\textbf{SG: #1}}}
\newcommand{\shnote}[1]%
    {\textcolor{blue}{\textbf{#1}}}
\newcommand{\remove}[1]%
    {\textcolor{red}{#1}}

\begin{document}

\title{Reachability-Based Safety Guarantees using Efficient Initializations}

\author{

\authorblockN{
Sylvia L. Herbert, Shromona Ghosh, Somil Bansal, and Claire J. Tomlin 
}
}

\maketitle
\begin{abstract}
Hamilton-Jacobi-Isaacs (HJI) reachability analysis is a powerful tool for analyzing the safety of autonomous systems.  This analysis is computationally intensive and typically performed offline.  Online, however, the autonomous system may experience changes in system dynamics, external disturbances, and/or the surrounding environment, requiring updated safety guarantees. 
Rather than restarting the safety analysis, we propose a method of ``warm-start'' reachability, which uses a user-defined  initialization (typically the previously computed solution). By starting with an HJI function that is closer to the solution than the standard initialization, convergence may take fewer iterations. 

In this paper we prove that warm-starting will result in guaranteed conservative solutions by over-approximating the states that must be avoided to maintain safety. We additionally prove that for many common problem formulations, warm-starting will result in exact solutions.
We demonstrate our method on several illustrative examples with a double integrator, and also on a more practical example with a 10D quadcopter model that experiences changes in mass and disturbances and must update its safety guarantees accordingly. We compare our approach to standard reachability and a recently proposed ``discounted'' reachability method, and find for our examples that warm-starting is $1.6$ times faster than standard and $6.2$ times faster than (untuned) discounted reachability.\vspace{-2em}
\end{abstract}

\IEEEpeerreviewmaketitle

\input{intro.tex}
\newgeometry{top=.75in,bottom=.78in,right=0.75in,left=.75in}
\input{formulation.tex}

\input{warmstart.tex}

\input{examples.tex}

\input{fancyexample.tex}

\input{conclusion.tex}


\input{appendix.tex}

\bibliographystyle{unsrtnat}
\bibliography{references}
\end{document}

%% file: notation.tex
\newcommand{\tvar}{t}
\newcommand{\tdummy}{\tau}
\newcommand{\thor}{0} 
\newcommand{\R}{\mathbb{R}}
\newcommand{\ctrl}{u}
\newcommand{\dstb}{d}
\newcommand{\cfunc}{u(\cdot)}
\newcommand{\dfunc}{d(\cdot)}
\newcommand{\cset}{\mathcal{U}}
\newcommand{\cfset}{\mathbb{U}}
\newcommand{\dset}{\mathcal{D}}
\newcommand{\dfset}{\mathbb{D}}
\newcommand{\state}{x}
\newcommand{\traj}{\xi} 

\newcommand{\dyn}{f} 
\newcommand{\targetfunc}{l}
\newcommand{\targetset}{\mathcal{L}}
\newcommand{\costfunctional}{J}

\newcommand{\vfunc}{V}
\newcommand{\vset}{\mathcal{V}}
\newcommand{\vfuncl}{\vfunc_{\targetfunc}}
\newcommand{\vconv}{\vfunc^*}
\newcommand{\warmfunc}{k}
\newcommand{\warmset}{\mathcal{K}}
\newcommand{\vfunck}{\vfunc_{\warmfunc}}
\newcommand{\trajstandard}{\traj_{\state,\tvar}^{\ctrl,\dstb}}

\newtheorem{remark}{Remark}

\newtheorem{theorem}{Theorem}
\newtheorem{proposition}{Proposition}

\newtheorem{lemma}{Lemma}

\newcommand{\runningexample}[1]%
{
\textbf{Running example:}
\textit{#1}
}

%% file: macros.tex
\newcommand{\controlspace}{\mathcal{U}}
\newcommand{\disturbancespace}{\mathcal{D}}

%% file: intro.tex
\blfootnote{This research is supported by an NSF CAREER award, Sylvia Herbert's NSF GRFP, NSF's CPS FORCES and VeHICaL projects, the UC-Philippine-California Advanced Research Institute, the ONR MURI Embedded Humans, the DARPA Assured Autonomy program, and the SRC CONIX Center.\\ All authors are with the Department of Electrical Engineering and Computer Sciences at UC Berkeley. Contact info: \{\href{mailto:sylvia.herbert@berkeley.edu}{sylvia.herbert},
\href{mailto:shromona.ghosh@berkeley.edu}{shromona.ghosh},
\href{mailto:sylvia.herbert@berkeley.edu}{somil},
\href{mailto:tomlin@berkeley.edu}{tomlin}\}@berkeley.edu}

\section{Introduction}
\label{sec:intro}
As humanity increasingly relies on autonomous systems, ensuring provable safety guarantees and controllers for these systems is vital. To achieve safety for nonlinear systems, tools such as Hamilton-Jacobi-Isaacs (HJI) reachability analysis can provide both a guarantee and a corresponding control input \cite{bansal2017hamilton,Mitchell2005}.  Applications include collision avoidance \cite{Mitchell2005,chen2016robust}, safe tracking of online motion planners \cite{herbert2017fastrack,kousik2018bridging}, stormwater management \cite{chapman2018reachability}, and administering anesthesia \cite{kaynama12}. 
HJI reachability analysis is based on assumptions about system dynamics, external disturbances, and the surrounding environment. However in reality the dynamics, the disturbance bounds, or the environment may differ from the assumptions. In these situations the safety analysis must be adapted. 

Unfortunately, performing HJI reachability analysis is computationally intensive for large systems and cannot be computed efficiently as new information is acquired, thus and area of current research interest is developing frameworks for efficiently updating the computation as new information is acquired.  There are some methods for speeding up this computation  using decomposition \cite{Chen2016DecouplingJournal}, and there are other efficient approaches that require simplified problem formulations and/or dynamics \cite{frehse11, greenstreet98, kurzhanski00, kurzhanski02, maidens13,nilsson16, kim2019toolbox}.  
The methods in \cite{althoff15, chen13, dreossi16, frehse11,majumdar14}, can handle more complex dynamics, but may be less scalable or unable to represent complex sets. Efficient reachability analysis remains challenging for general system dynamics and problem setups.

Warm-starting in the optimization community involves using a initialization that acts as a ``best guess'' of the solution, and therefore may converge in fewer iterations (if convergence can be achieved). Recent work applied this warm-starting idea to create a ``discounted reachability'' formulation for infinite-time horizon problems \cite{akametalu2018minimum,fisac2019bridging}.  By using a discount factor, this formulation \textit{guarantees} convergence regardless of the initialization.  
However, due to this discount factor convergence rates can be very slow, and in practice the analysis may not converge numerically when convergence thresholds are too tight, or may converge incorrectly when convergence thresholds are too lenient.  In addition, parameter tuning of the discount factor can be time-intensive.  These issues reduce the computational benefit of warm-start reachability.

Until now there were no guarantees of convergence for warm-starting HJ reachability without using a discount factor. In this paper we prove that warm-start reachability with no discount factor will in general result in \textit{guaranteed conservative} safety analyses and controllers (i.e. the analysis over-approximates the set of states that are unsafe to enter).
Moreover, if the initialization is over-optimistic and therefore dangerous (i.e. the initialization underestimates the set of states that are unsafe to enter), we prove that warm-starting is \textit{guaranteed to converge exactly} to the true solution (here we use ``exact''  to mean numerically convergent \cite{mitchell2004toolbox}).

In addition to these proofs, we provide several common problem classes for which we can prove this exact convergence. We demonstrate these results on an illustrative example with a double integrator, and a more practical example of a realistic 10D quadcopter model that experiences changes in mass and disturbances and must update its safety guarantees accordingly. In these examples warm-start reachability is $1.6$ times faster than standard reachablity and $6.2$ times faster than (untuned) discounted reachability formulation. 


%% file: formulation.tex
\section{Problem Formulation}
\label{sec:formulation}
Consider an autonomous agent in an environment in the presence of external disturbance.  This environment contains a target set $\targetset$ that is meaningful to the agent: it can be either a set of goal states, or a set of unsafe states.  HJI reachability seeks to find the set of initial states for which the system acting optimally and under worst-case disturbances will end up in the target set $\targetset$ either at a particular time (backward reachable set, or BRS) or within a time horizon (backward reachable tube, or BRT). Optimal behavior of the system depends on the nature of the target set and can be formulated as a game: for a goal set, the control will seek to minimize distance to the goal whereas the worst-case disturbance will maximize distance to the goal.  For an unsafe set, the control will maximize and the disturbance will minimize. Both cases can be solved using HJI reachability analysis. Note that there are reach-avoid problems that seek to reach a goal set while avoiding unsafe sets.  There are also problems that use \textit{forward} reachable sets and tubes. More information can be found in \cite{bansal2017hamilton}.

The theory in this paper applies to BRTs with infinite-time horizons.  Typically this scenario is more interesting in the avoid case (where the system seeks to avoid an unsafe set of states forever), and will therefore be the focus of this paper. 
In this section we define the agent's dynamics and formally introduce HJI reachability analysis.

\subsection{Dynamic System Model}

We assume that the autonomous system (i.e. agent) has initial state  $\state \in \R^n$ and initial time $\tvar$, and evolves according to the ordinary differential equation (ODE):

\begin{equation}
\label{eq:dyn}
\begin{aligned}
\dot{\state} = \dyn(\state, \ctrl, \dstb), \quad \ctrl \in \cset, \dstb \in \dset.
\end{aligned}
\end{equation}

\noindent Here the system has a control $\ctrl$ and disturbance $\dstb$. We assume that these inputs are drawn from compact sets ($\cset$, $\dset$), and their signals over time ($\ctrl(\cdot)$, $\dstb(\cdot)$) are drawn from the set of measurable functions $\cfset:[\tvar,\thor]\rightarrow \cset$ and  $\dfset:[\tvar,\thor]\rightarrow \dset$. 

We assume that the flow field $\dyn: \mathbb{R}^{n}\times\cset\times\dset \rightarrow \mathbb{R}^{n}$ is uniformly continuous in time and Lipschitz continuous in $\state$ for fixed $\ctrl$ and $\dstb$. Under these assumption there exists a unique solution of these system dynamics for a given $\cfunc, \dfunc$~\cite{EarlA.Coddington1955}, providing trajectories of the system:
$\traj(\tdummy; \state, \tvar, \cfunc, \dfunc).$
This notation can be read as the state achieved at time $\tdummy$ by starting at initial state $\state$ and initial time $\tvar$, and applying input functions $\cfunc$ and $\dfunc$ over $[\tvar,\tau]$.  
For compactness we will refer to trajectories using $\trajstandard(\tdummy).$ Because we tend to solve reachability problems backwards in time, we use the notation that forward trajectories end at final time $\tdummy = \thor$, and start at an initial negative time $\tvar$. 

\runningexample{In this paper we use double integrator as a running example with dynamics,
\begin{align}
\label{eqn:doubleint}
\dot \state = \begin{bmatrix}
\dot p \\
\dot v
\end{bmatrix} = \begin{bmatrix}
v + \dstb\\
\ctrl b
\end{bmatrix},
\end{align}
\noindent with states position $p$ and velocity $v$, where $\ctrl\in[-1, 1]$ is acceleration. By default the disturbance is $\dstb=[0,0]$, and there is a default model parameter of $b = 1$. In later examples we will change the disturbance bound and model parameter.}

\subsection{Hamilton-Jacobi-Isaacs Reachability}
\subsubsection{Defining the Value Function}
We define a target function $\targetfunc(\state)$ whose subzero level set is the target set $\targetset$ describing the unsafe states, i.e. $\targetset = \{\state : \targetfunc(\state) \leq 0\}$.  Typically  $\targetfunc(\state)$ is defined as a signed distance function that measures distance to $\targetset$. This can be considered as measure of reward, with positive reward outside of the unsafe set and negative reward inside. 

This problem formulation seeks to find all trajectories that will enter $\targetset$ at any point in the time horizon, and therefore become unsafe:\vspace{-.5em}

\begin{equation}
    \label{eq:costfunctional}
    \costfunctional(\state,\tvar,\cfunc,\dfunc) = \min_{\tdummy \in [\tvar,0]} \targetfunc(\trajstandard(\tdummy)), ~ \tvar \leq 0.
\end{equation}

More specifically, the goal is to capture this minimum reward for \textit{optimal trajectories} of the system.  To do this we optimize for the optimal control signal that maximizes the reward (and drives the system away from the unsafe target set) and the worst-case disturbance signal that minimizes the reward. This leads to the value function:\vspace{-.5em}

\begin{equation}
    \label{eq:valuefunc}
    \vfunc(\state,\tvar) = \inf_{\gamma[\cfunc](\cdot)} \sup_{\cfunc} \Big\{\costfunctional\Big(\state,\tvar,\cfunc, \gamma[\cfunc](\cdot)\Big)\Big\}.
\end{equation}

\noindent Note that the disturbance is $\gamma = \{\dstb : \cset \rightarrow \dset\}$, which maps control inputs to disturbance inputs.  As in \cite{Mitchell2005}, we restrict the disturbance to draw from nonanticipative strategies.

Level sets of the value function correspond to level sets of the target function. If a state has a negative value, optimal trajectories starting from that state 
have entered the target set $\targetset$ sometime within the time horizon.  Therefore, the sub-zero level set of the value function comprises the backward reachable tube (BRT), notated as $\vset$: the set of states from which the system is guaranteed to enter the target set within the time horizon under optimal control and worst-case disturbance. For the infinite-time avoid BRT, we define the converged value function as  $\vconv(\state) = \lim_{\tvar  \rightarrow -\infty}\vfunc(\state,\tvar)$.  The subzero level set of this converged value function is the infinite-time avoid backwards reachable tube: $\vset^* = \{\state: \vfunc^*(\state) \le 0\}$. Trajectories initialized from states in this set will eventually enter the unsafe target set despite the control's best effort.  The complement of this set is therefore the safe set.

\runningexample{In the running example the target set $\targetset = \{(p,v) : |p| \leq 2, |v|<\infty\}$.
The target function $\targetfunc(\state)$ and its corresponding target set $\targetset$ can be seen in Fig.~\ref{fig:original_results} in green.  The converged BRT $\vset^*$ and value function $\vconv(\state)$  are in cyan.  If the system starts inside $\vset^*$, it will eventually enter the unsafe target set even while applying the optimal control (i.e. decelerating/accelerating as much as possible).}

\begin{figure}
\centering
\includegraphics[width=\columnwidth]{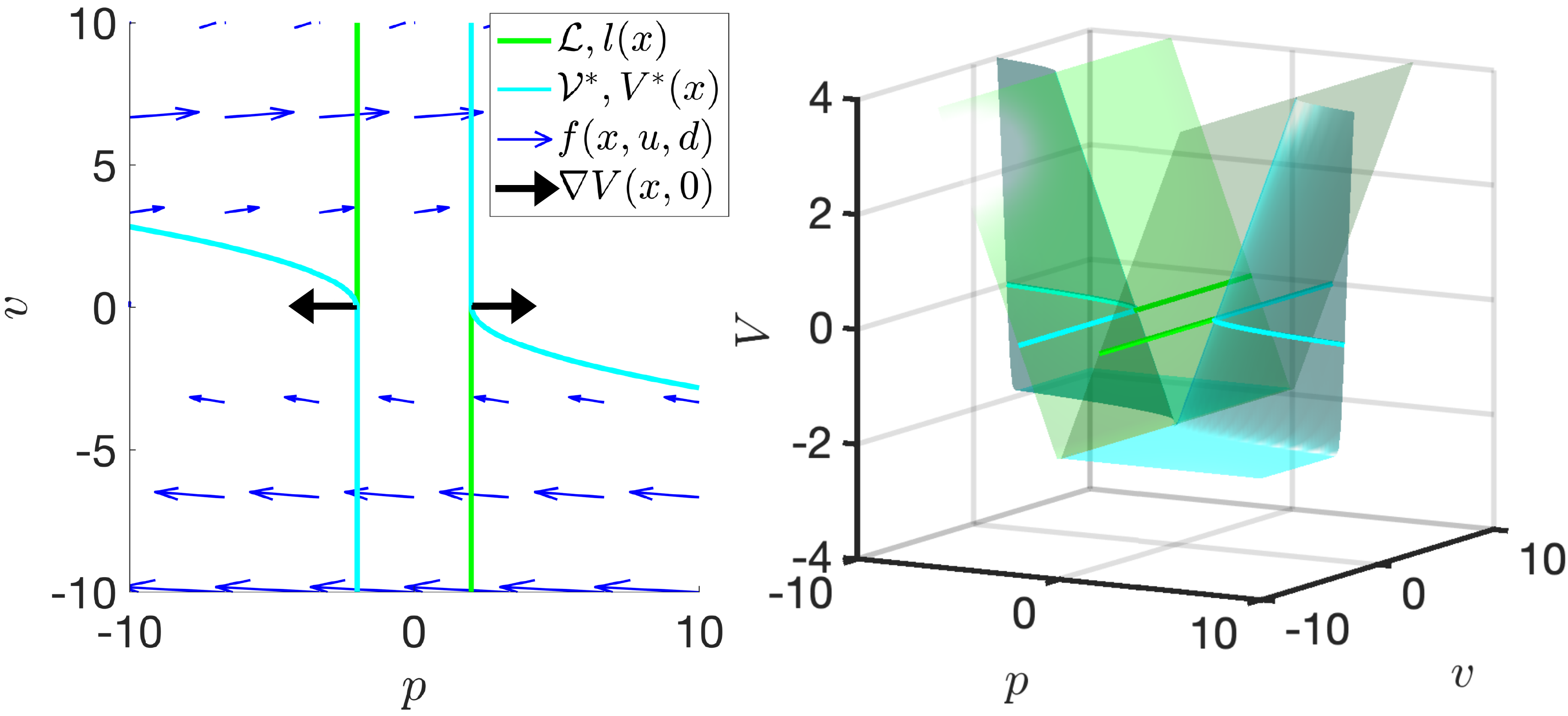}
\caption{Visualization of the running example using a double integrator model. The target set $\targetset$ and corresponding function $\targetfunc(\state)$ are in green. We initialize $\vfunc(\state,0)=\targetfunc(\state)$, and update the function using \eqref{eq:HJIVI} by optimizing over the inner product between the spatial gradients (seen for $\vfunc(\state,0)$ as black arrows) and the system dynamics (whose flow field is seen as blue arrows). The converged BRT $\vset^*$ and value function $\vfunc^*(\state)$ are in cyan.}
\vspace{-2em}
\label{fig:original_results}
\end{figure}
\subsubsection{Solving for the Value Function}
To solve this optimization problem for the value function, we discretize the state space and initialize the value function to be equal to the target function, $\vfunc(\state,0) = \targetfunc(\state)$.  We work backwards in time while updating the value function until the initial time has been reached (or in the infinite-horizon case, until convergence).
The change in $\vfunc(\cdot, \cdot)$ for each state backwards in time satisfies the Hamilton-Jacobi-Isaacs partial differential equation (HJI PDE)~\eqref{eq:HJIPDE}, which is used for solving backward reachable sets (BRSs): the set of initial states that will lead to entering the target set at \textit{exactly} $\tdummy = \thor$:
\begin{equation}
    \label{eq:HJIPDE}
    \begin{aligned}
    D_\tvar \vfunc(\state,\tvar)+ H\Big(\vfunc(\state,\tvar), \dyn(\state,\ctrl,\dstb)\Big) = 0.
    \end{aligned}
\end{equation}
where $H$ is the Hamiltonian defined in~\eqref{eq:ham} optimizes over the inner product between the spatial gradients of the value function and the flow field of the dynamics to compute the optimal control and disturbance inputs.
\begin{equation}
    \label{eq:ham}
    \begin{aligned}
    H\Big(\vfunc(\state,\tvar), \dyn(\state,\ctrl,\dstb)\Big) =&\\ \max_\ctrl \min_\dstb & \langle \nabla \vfunc(\state,\tvar), \dyn(\state,\ctrl,\dstb)\rangle.
        \end{aligned}
\end{equation}

\runningexample{For the running example the initial spatial gradients for $\vfunc(\state,0) = \targetfunc(\state)$ can be seen as black arrows in Fig. \ref{fig:original_results}. The Hamiltonian \eqref{eq:ham} will optimize over the inner product between these gradients and the flow field of the dynamics $\dyn(\state,\ctrl,\dstb)$, seen as blue arrows.}

Note that the HJI PDE \eqref{eq:HJIPDE} solves for $\costfunctional(\state,\tvar,\cfunc,\dfunc) =  \targetfunc(\trajstandard(\thor))$.
Because we are solving a BRT and want to capture the minimum value over the \textit{entire} time horizon (as in \eqref{eq:costfunctional}), we must include a minimization term, converting the HJI PDE to an HJI variational inequality (HJI VI).
\begin{equation}
\begin{aligned}
    \label{eq:HJIVI}
    \min\Big\{D_\tvar \vfunc(\state,\tvar)+H\Big(\vfunc(\state,\tvar),& \dyn(\state,\ctrl,\dstb)\Big),\\
    &\targetfunc(\state)-\vfunc(\state,\tvar)\Big\} = 0.
    \end{aligned}
\end{equation}

The term $\targetfunc(\state)-\vfunc(\state,\tvar)$ restricts the value function from becoming more positive than the target function, effectively enforcing that all trajectories that achieve negative reward at any time will continue to have negative reward for the rest of the time horizon. For more details on the derivation of this HJI VI and variations that include forward reachability and reach-avoid scenarios, please refer to \cite{margellos2011hamilton,fisac2015reach,bansal2017hamilton,Mitchell2005}.
We solve this HJI VI recursively using dynamic programming:\vspace{-.5em}
\begin{equation}
\label{eq:dynprogramming}
    \begin{aligned}
        \vfunc(\state,\tvar) = \max_{\cfunc}\min_{\dfunc}\min\Big\{&
        \inf_{\tdummy\in[\tvar,\tvar+dt)} \targetfunc\Big(\trajstandard(\tdummy)\Big),\\
        &\vfunc\Big(\trajstandard(\tvar+dt),\tvar+dt\Big)
        \Big\}.
        \end{aligned}
    \end{equation}
We use \eqref{eq:dynprogramming} to update the value of each discretized state backwards in time using the level set method toolbox and associated helperOC toolbox \cite{mitchell2004toolbox, chen2017helperOC}. At convergence the result is the infinite-horizon value function, whose subzero level set $\vset^*$ corresponds to the set of states that should be avoided in order to remain safe for all time.  Online the system avoids these states by solving for the instantaneous optimal control at state $\state$ using the Hamiltonian and inifinite-horizon value function:
\begin{equation}
    \label{eq:uOpt}
    \ctrl^*= \arg\!\max_\ctrl \min_\dstb \langle\nabla \vconv(\state), \dyn(\state,\ctrl,\dstb)\rangle.
\end{equation}
\subsection{Discounted Reachability}
In~\cite{akametalu2018minimum} and~\cite{fisac2019bridging}, the authors introduced a discounting factor $\lambda$ into the cost function~\eqref{eq:costfunctional} motivated by the sum of discounted rewards in reinforcement learning. Introducing $\lambda$, now updates the dynamic programming in~\ref{eq:dynprogramming} to solve, 
\begin{equation}
    \begin{aligned}
        \vfunc(\state,\tvar) = \max_{\cfunc}\min_{\dfunc}\min\Big\{&
        \inf_{\tdummy\in[\tvar,\tvar+dt)} \targetfunc\Big(\trajstandard(\tdummy)\Big)\exp{(\lambda\cdot \tau)},\\
        &\exp{(\lambda\cdot dt)}\vfunc\Big(\trajstandard(\tvar+dt),\tvar+dt\Big)
        \Big\},
    \end{aligned}
\end{equation}
where $\tvar \le 0$.
The authors show that this is a contraction mapping and, hence, $\vfunc(\state,0)$ can be initialized to any function (not just $\targetfunc(\state)$). 
The discounting allows $\vfuncl(\state, \tvar)$ to \textit{forget} any incorrect initializations over longer time horizon.
However, this formulation can still result in a slow convergence without careful tuning of $\lambda$, as we demonstrate in Section \ref{sec:exact_examples}.

%% file: warmstart.tex
\section{Warm-Start Reachability}
\label{sec:warmstartConv}
When there are minor changes to the problem formulation, such as changes to the model parameters, external disturbances, or target sets, computing $\vconv(\cdot, \cdot)$ requires recomputing the entire value function starting with the target function ($\vfunc(\state,0)=\targetfunc(\state)$).  Instead, we initialize with a previous computed (converged) value function. 
We define this warm-starting function as $\warmfunc(\state)$, with subzero level set $\warmset = \{\state: \warmfunc(\state) < 0\}$.

To develop the theory, we revisit the cost function~\eqref{eq:costfunctional}. We rewrite, 
$\costfunctional_{\targetfunc}(\state, \tvar, \ctrl(\cdot), \dstb(\cdot)) = \min\Big\{\inf_{\tdummy \in [\tvar,0)} \targetfunc(\trajstandard(\tdummy)), \targetfunc(\trajstandard(0))\Big \}$ and $\vfuncl(\state, \tvar)$ is defined as in~\eqref{eq:valuefunc} by replacing $\costfunctional$ by $\costfunctional_{\targetfunc}$,
\begin{equation}
    \label{eq:vl}
    \begin{split}
    \vfuncl(\state,\tvar) &= \max_{\cfunc}\min_{\dfunc} \costfunctional_{\targetfunc}(\state, \tvar, \ctrl(\cdot), \dstb(\cdot)) \\
    &= \max_{\cfunc}\min_{\dfunc} \min \Big\{
     \inf_{\tdummy \in [\tvar,0)} \targetfunc(\trajstandard(\tdummy)),\vfuncl(\trajstandard(0),0)\Big\},\\
    &= \max_{\cfunc}\min_{\dfunc} \min \Big\{
     \inf_{\tdummy \in [\tvar,0)} \targetfunc(\trajstandard(\tdummy)),\targetfunc(\trajstandard(0))\Big\},
    \end{split}
\end{equation}

$\vfuncl(\state, \tvar)$ is the solution to the following HJI-VI,
\begin{equation}
\label{eq:v_HJIVI}
    \begin{split}
        0=\min\Big\{&D_\tvar \vfuncl(\state,\tvar)+H_\targetfunc\Big(\vfuncl(\state,\tvar),\dyn(\state,\ctrl,\dstb)\Big),\\
        &\targetfunc(\state)-\vfuncl(\state, \tvar)\Big\},\\
        H_\targetfunc(\vfuncl(\state,\tvar)&,\dyn(\state,\ctrl,\dstb)) = \\
        &\max_u \min_d <\nabla \vfuncl(\state,\tvar),\dyn(\state,\ctrl,\dstb)>, \\
        \vfuncl(\state,0) &= \targetfunc(\state,0).
    \end{split}
\end{equation}
The converged value function is defined as, $\vconv_\targetfunc(\state) = \lim_{\tvar  \rightarrow -\infty}\vfunc_\targetfunc(\state,\tvar)$.

When we warm-start the computation of value function using $\warmfunc$, the cost function is given by,
\begin{equation}
    \label{eqn:warm_cost}
    \costfunctional_{\warmfunc}(\state, \tvar, \ctrl(\cdot), \dstb(\cdot)) = \min\Big\{\inf_{\tdummy \in [\tvar,0)} \targetfunc(\trajstandard(\tdummy)), \warmfunc(\trajstandard(0))\Big \}
\end{equation}
$\vfunck$ can be defined as in~\eqref{eq:valuefunc} with $\costfunctional = \costfunctional_{\warmfunc}$, i.e,
\begin{equation}
    \label{eq:vk}
    \begin{split}
    \vfunck(\state,\tvar) &= \max_{\cfunc}\min_{\dfunc} \costfunctional_{\warmfunc}(\state, \tvar, \ctrl(\cdot), \dstb(\cdot)) \\
    &=\max_{\cfunc}\min_{\dfunc} \min \Big\{
     \inf_{\tdummy \in [\tvar,0)} \targetfunc(\trajstandard(\tdummy)),\warmfunc(\trajstandard(0))\Big\}.
    \end{split}
\end{equation}
$\vfunck$ is the solution to the HJI-VI defined similarly as in~\eqref{eq:v_HJIVI} with $\vfunck(\state, 0) = \warmfunc(\state)$.
The converged value function is defined as, $\vconv_\warmfunc(\state) = \lim_{\tvar  \rightarrow -\infty}\vfunc_\warmfunc(\state,\tvar)$.

In this section we prove that the converged value function $\vconv_\warmfunc(\state)$ that is initialized as above $\vfunck(\state,0) = \warmfunc(\state)$ will always be more negative than the value function $\vconv_\targetfunc(\state)$ achieved by standard reachability (i.e. initialized as $\vfuncl(\state,0)=\targetfunc(\state)$).  For the case of avoiding an unsafe set, this means that the relationship between the functions' BRTs (i.e. subzero level sets) is $\vset_\warmfunc^*\supseteq \vset_\targetfunc^*$. In other words, $\vset_\warmfunc^*$ is a conservative over-approximation of $\vset_\targetfunc^*$.  We will prove that for certain conditions (more explicitly, when $\warmfunc(\state)\geq\vconv_\targetfunc(\state)$), we can guarantee that the resulting value function and BRT will be exact.

\subsection{Conservative Warm-Start Reachability}
If $[\vfunck(\state,0) = \warmfunc(\state)]\leq\vconv_\targetfunc(\state)$, a contraction mechanism is required to raise $\vfunck^*(\state)$ towards the true solution $\vconv_\targetfunc(\state)$.  Recall the HJI VI from \eqref{eq:HJIVI}. Contraction may happen naturally, when the left hand side of the minimization (the HJI PDE) ``pulls the system up'' due to the Hamiltonian.  However, there are no guarantees that this contraction will happen, and the new value function may get stuck in a local solution, $\vfunck^*(\state)\leq \vfuncl^*(\state)$. This will result in a conservative BRT.

\begin{theorem}
\label{thm:allk}
For all initializations of $\vfunck(\state, 0) = \warmfunc(\state)$, the result will be conservative, i.e. 
\begin{equation}
    \forall \state, \tvar<0, \; \vfunck(\state, \tvar) \leq \vfuncl(\state, \tvar) 
\end{equation}
\end{theorem}
\begin{proof}
We prove that $\vfunck(\state, \tvar) \leq \vfuncl(\state, \tvar)$ for two cases, (a) $\warmfunc(\state) < \targetfunc(\state)$ and (b) $\warmfunc(\state) \geq \targetfunc(\state)$. 
\\
\paragraph{ $\warmfunc(\state) < \targetfunc(\state)$}
For $\forall \state, \tvar < 0$, let $\vfuncl(\state, \tvar)$ be defined as~\eqref{eq:vl} and $\vfunck(\state, \tvar)$ be defined as~\eqref{eq:vk}.
At $\tvar=0$, we have $\Big[\vfunck(\state, 0) = \warmfunc(\state)\Big] < \Big[\targetfunc(\state) = \vfuncl(\state, 0)\Big] \Rightarrow \vfunck(\state, 0) < \vfuncl(\state, 0)$. 
For any $\tvar < 0$:
\begin{equation}
    \label{eq:time}
    \begin{aligned}
    \vfunck(\state,\tvar) &= \max_{\cfunc}\min_{\dfunc} \min \Big\{
     \inf_{\tdummy \in [\tvar,0)}\targetfunc(\trajstandard(\tdummy)), \warmfunc( \trajstandard(0))\Big\},\\
    &\leq \max_{\cfunc}\min_{\dfunc} \min \Big\{
     \inf_{\tdummy \in [\tvar,0)}\targetfunc(\trajstandard(\tdummy)), \targetfunc(\trajstandard(0))\Big\},\\
    &= \vfuncl(\state,\tvar).
    \end{aligned}
\end{equation}
The second inequality follows from the fact that $\warmfunc(\state) < \targetfunc(\state)~\forall \state \in \R^{n}$.
Hence, $\forall \state, \tvar$, we have $\vfunck(\state, \tvar) \leq \vfuncl(\state, \tvar)$.

Finally, $\tvar \rightarrow -\infty$, we have $\vconv_\warmfunc(\state) \leq \vconv_\targetfunc(\state)$.
\\
\paragraph{ $\warmfunc(\state) \geq \targetfunc(\state)$}
When $\tvar = 0$, $\Big[\vfunck(\state, 0) = \warmfunc(\state)\Big] \geq\Big[ \targetfunc(\state) = \vfuncl(\state, 0)\Big]$. 
For a time instance $\tvar = 0^-$,
\begin{equation}
\begin{aligned}
    \vfunck(\state, \tvar) &= \max_{\cfunc}\min_{\dfunc} \min \Big\{\inf_{\tdummy \in [\tvar,0)} \targetfunc(\trajstandard(\tdummy)),\warmfunc(\trajstandard(0))\Big\}\\
    &= \min\{\targetfunc(\trajstandard(0^-)), \warmfunc(\trajstandard(0)\} \\
    &= \targetfunc(\trajstandard(0^-)) = \vfuncl(\state, \tvar).
\end{aligned}
\end{equation} 
We can re-write~\eqref{eq:vl} and~\eqref{eq:vk} by replacing $0$ by $0^-$. The rest follows from proof of case (a). 
Here $0^-$ implies an infinitesimally small change in time and we are effectively computing $\vfunck(\state, 0^-) = \min(\warmfunc(\state), \targetfunc(\state))$ and treating $\vfunck(\state, 0) = \vfunck(\state, 0^-)$. One could derive the same proof by considering $\vfunck(\state, 0) = \min(\warmfunc(\state), \targetfunc(\state))$.
\end{proof}
In other words, the converged warm-starting solution will never be \textit{more} conservative than the initialization, and at least as conservative as the exact solution.

\subsection{Exact Warm-start Reachability}
In the case in which $\warmfunc(\state)\geq \vconv_\targetfunc(\state)$, we are additionally guaranteed to recover the \textit{exact} solution.
\begin{theorem}
\label{theorem:equality}
If we warm-start with $\vfunck(\state, 0) = \warmfunc(\state)$, such that $\forall \state \; \warmfunc(\state) \geq \vconv_{\targetfunc}(\state)$, then,
\begin{equation}
    \vconv_\warmfunc(\state) = \vconv_\targetfunc(\state)
\end{equation}
\end{theorem}
The proof of Theorem~\ref{theorem:equality} follows from Theorem~\ref{thm:allk} and Lemma~\ref{lem:kgreaterV}, defined as:
\begin{lemma}
\label{lem:kgreaterV}
If $\forall \state \; \warmfunc(\state) \geq \vconv_{\targetfunc}(\state)$, we have,
\begin{equation}
   \vfunck(\state, \tvar) \geq \vconv_{\targetfunc}(\state)\quad
    \forall \state, \tvar <0
\end{equation}
\end{lemma}
\begin{proof}
To prove Lemma~\ref{lem:kgreaterV}, let us consider $\warmfunc'(\state) = \vconv_{\targetfunc}(\state)$. 
For $\tvar < 0$, using dynamic programming we have, 
\begin{equation}
    \begin{split}
    \vfunc_{\warmfunc'}(\state,\tvar) &= \max_{\cfunc}\min_{\dfunc} \min \Big\{
     \inf_{\tdummy \in [\tvar,0)} \targetfunc(\trajstandard(\tdummy)),\vfunc_{\warmfunc'}(\trajstandard(0),0)\Big\},\\
    &=\max_{\cfunc}\min_{\dfunc} \min \Big\{
     \inf_{\tdummy \in [\tvar,0)} \targetfunc(\trajstandard(\tdummy)),\warmfunc'(\trajstandard(0))\Big\}.
    \end{split}
\end{equation}
For any $\tvar < 0$, we can follow the same logic as in \eqref{eq:time}. 
Hence, $\forall \state, \tvar$, we have $\vfunc_{\warmfunc'}(\state, \tvar) \leq \vfunck(\state, \tvar)$.
Moreover, since $\vfunc_{\warmfunc'}(\state, 0) = \vconv_{\targetfunc}(\state)$, we know that $\vfunc_{\warmfunc'}(\state, \tvar) = \vconv_{\targetfunc}(\state)\quad \forall \tvar$ since $\vconv_{\targetfunc}(\state)$ is the converged value function corresponding to $\targetfunc(\state)$. Hence, $\vfunck(\state, \tvar) \geq \vconv_{\targetfunc}(\state) \quad \forall \state, \tvar < 0$.
Since this holds for all time, it also holds for $\tvar \rightarrow -\infty$: $\vconv_\warmfunc(\state) \geq \vconv_\targetfunc(\state)$.
\end{proof}
\begin{figure*}
\centering
\includegraphics[width=1.8\columnwidth]{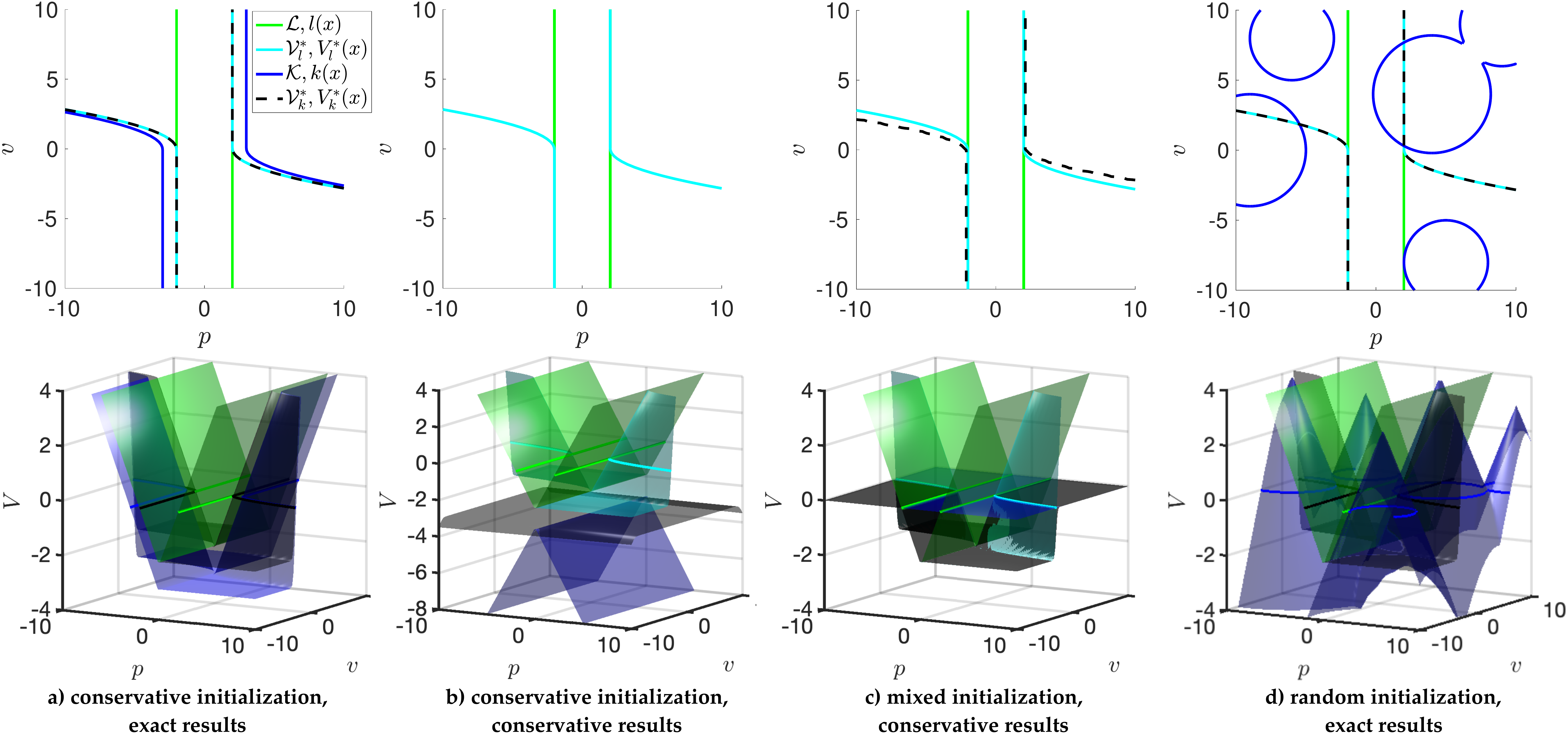}
\caption{The top row shows the target sets and backward reachable tubes, which are the subzero level sets of the target and value functions (bottom row). For all examples shown, green is the target set and function, cyan is the true BRT and converged value function, blue is the warm-start initialization, and black is the warm-start converged value function. (a) conservative warm-start initialization that converges exactly. (b) somewhat unrealistic conservative warm-start initialization that gets stuck in a local solution and results in a conservative value function ($\warmset$ and $\vset^*_\warmfunc$ are not visualized because they include the entire state space). (c) initializing at zero everywhere ($\warmset$ not visualized because it includes the entire state space) results in a slightly conservative BRT. (d) to demonstrate how well this algorithm works in practice, we initialize with the complement of random circles, resulting in exact convergence.}
\label{fig:conservative_results}
\vspace{-1.5em}
\end{figure*}
\begin{proof}
To prove Theorem~\ref{theorem:equality}, we have $\forall \state$, \vspace{-.5em}
\begin{equation}
    \begin{split}
        \vconv_\warmfunc(\state) &\leq \vconv_\targetfunc(\state) \quad \text(Theorem~\ref{thm:allk}) \\
        \vconv_\warmfunc(\state) &\geq \vconv_\targetfunc(\state) \quad \text(Lemma~\ref{lem:kgreaterV}) \\
        \Rightarrow \vconv_\warmfunc(\state) &= \vconv_\targetfunc(\state)
    \end{split}
    \vspace{-.5em}
\end{equation}
\end{proof}
\vspace{-1.5em}

%% file: examples.tex
\section{Conservative Warm-Start Examples}
\label{sec:examples_conserv}
Below we demonstrate several scenarios using the running example that result in conservative solutions.  All experiments were run on a desktop computer with an Intel Core i7-5280K CPU $@ 3.30$GHz $\times 12$ processor and $12.8$GB of memory. For all examples the value function is considered converged when the maximum change of value in one time step ($dt = 0.01$) is less than $0.001$.

\subsection{Conservative Initialization with Exact Results}
In practice we find that frequently the value function converges to the exact solution even when initialized below the converged value function, i.e. when $\warmfunc(\state) < \vconv_\targetfunc(\state)$.  Fig. \ref{fig:conservative_results}a demonstrates one such example. The warm-start function $\warmfunc(\state)$ (seen in blue) is initialized to be the original value function acquired when $\ctrl\in[-.7,.7]$.  If the control authority increases to $\ctrl\in[-1,1]$, standard reachability converges to the cyan value function $\vconv_\targetfunc(\state)$. In black is the value function under $\vconv_\warmfunc(\state)$ that was initialized by $\warmfunc(\state)$ instead of $\targetfunc(\state)$. Convergence occurs due to the Hamiltonian in \eqref{eq:HJIVI} contracting the value function until the solution has been reached.

\subsection{Conservative Initialization with Conservative Results}
To find a result that does not converge exactly and instead results in a conservative solution, we initialize with $\warmfunc(\state)<\vconv_\targetfunc(\state)$ that has incorrect gradients everywhere, as shown in blue in Fig \ref{fig:conservative_results}b.  Note that this is a fairly unrealistic initial estimate for the true value function, as as the subzero level set $\warmset$ is the entire state space.  As the Hamiltonian contracts the function, convergence occurs at a local solution when the gradients of the value function approach zero.  In black we see that $\vconv_\warmfunc(\state)<\vconv_\targetfunc(\state)$, and the BRT $\vset_\warmfunc^*$ is the entire state space.

\subsection{Mixed Initalization with Conservative Results}
In Fig.~\ref{fig:conservative_results}c we initialize the warm-starting function as $[\vfunck(\state,0)=\warmfunc(\state)]=0$ (blue) so that $\warmfunc(\state)\geq \vconv_\targetfunc(\state)$ for a subset of the state space. Where $\warmfunc(\state) \geq \vconv_\targetfunc(\state)$ convergence is nearly exactly (black), with slight conservativeness introduced at the boundary of where $\warmfunc(\state) = \vconv_\targetfunc(\state)$.  Where $\warmfunc(\state) < \vconv_\targetfunc(\state)$ the warm-start solution remains flat at $\vconv_\warmfunc(\state)=0$.  The resulting BRT $\vset_\warmfunc^*$ is a slight over-approximation of $\vset_\targetfunc^*$.

\subsection{Random Initialization with Exact Results}
Though we are able to find cases that lead to conservative results, these cases are hard to come by. In almost all initializations the correct value function was achieved exactly.  Fig. \ref{fig:conservative_results}d demonstrates this by initializing $\vset_\warmfunc$ with randomly spaced and sized circles. Similar exact results were found for a variety of system dynamics and problem formulations.

\section{Exact Warm-Start Examples} \label{sec:exact_examples}
\begin{figure}
\centering
\includegraphics[width=.85\columnwidth]{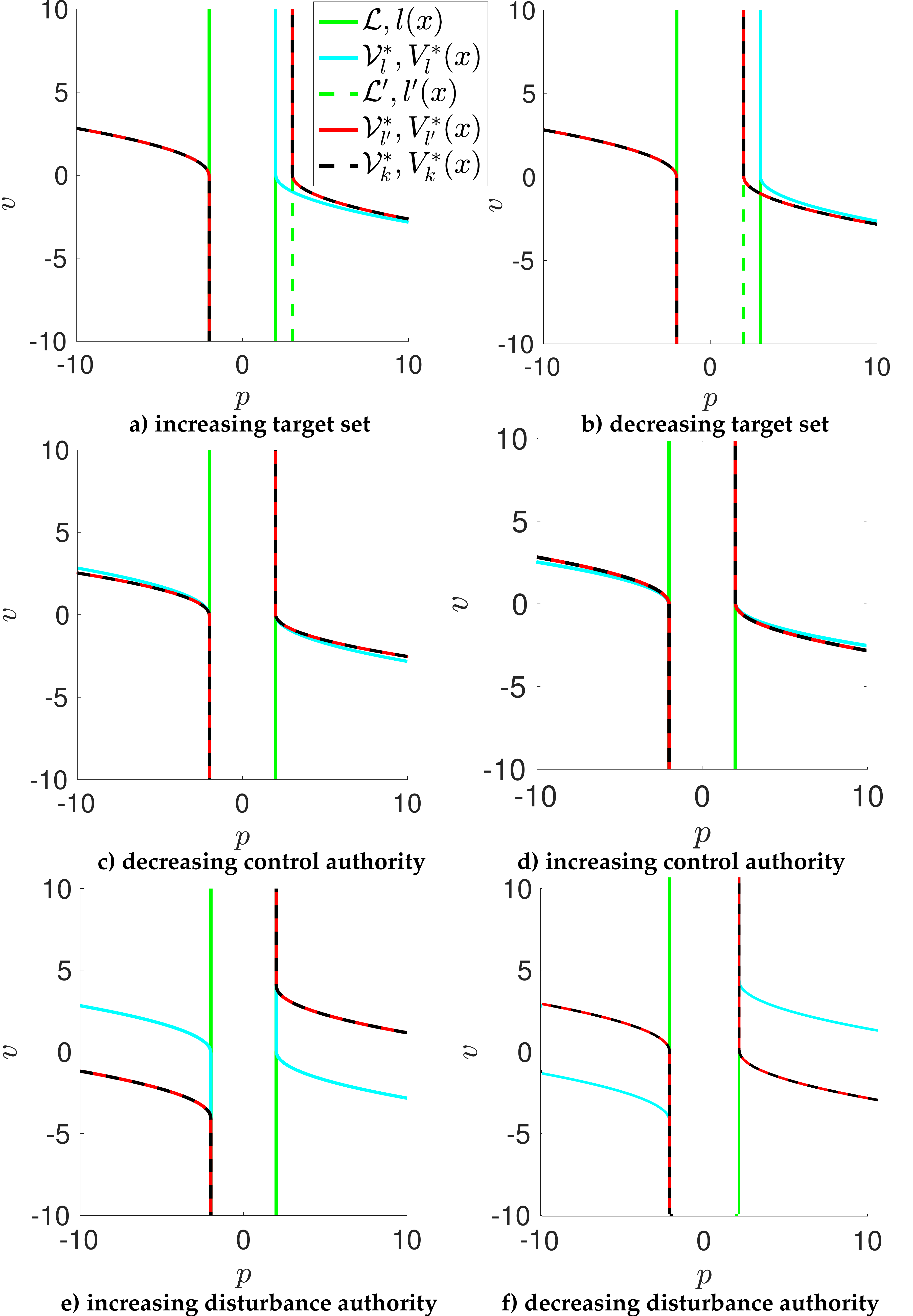}
\caption{For all examples shown, the region between the green lines is the target set. Similarly cyan marks the boundary of the original BRT, red marks the BRT based on new conditions, and black is the boundary of the warm-start converged BRT. The left column shows cases in which the exact solution (red) can be achieved by warm-starting (black) from a previous solution (cyan). The right column shows cases in which warm-starting (black) is guaranteed to at worst remain at the initialization (cyan) or at best will achieve the exact solution (red). In practice we generally achieve the exact solution.}
\label{fig:exact_results}
\vspace{-1.5em}
\end{figure}

\label{sec:ex_conserv}
Though in general we may not know if $\warmfunc(\state)\geq\vconv_\targetfunc(\state)$, there are some cases in which this can be proved, and therefore the exact solution can be recovered.  For all following examples $\vconv_\targetfunc(\state),$ is the original value function and $ \vset_\targetfunc^*$ is the corresponding BRT acquired from standard reachability using the default running example. Each subsection introduces changes to the problem formulations, resulting in a new $\vconv_{\targetfunc'}(\state),  \vset_{\targetfunc'}^{*}$ acquired from standard reachability. Finally, $\vconv_\warmfunc(\state), \vset^*_\warmfunc$ are the value function and BRT acquired by warm-starting with $\warmfunc(\state) = \vconv_\targetfunc(\state)$ with the changed problem formulation. We further show what happens when the conditions that lead to exact results are reversed. In these cases we cannot guarantee exact convergence, but can guarantee that in each iteration the function will either reduce conservativeness or remain in a local solution (i.e. $\warmfunc\leq \vfunck(\state,\tvar) \leq \vfuncl^*(\state)\; \forall \state, \tvar$).
We show in Table~1 a time comparison for each example to standard and discounted reachability, shown both in runtime and number of iteration steps. For the exact cases we find that warm-starting is consistently faster. For comparison to discounted reachability, we used a discount factor of 0.999 and annealed to a discount factor of 1 once convergence was reached (see \cite{akametalu2018minimum} for details).

\subsection{Changing Target Set}
\indent When the target set increases
($\targetset' \supseteq \targetset$), 
setting the initialization to the previously converged value results in $[\warmfunc(\state) = \vconv_{\targetfunc}(\state)]\geq\vconv_{\targetfunc'}(\state)$ and therefore exact convergence is guaranteed. Refer to Proposition~\ref{prop:targetset} and its proof is in Appendix~\ref{subsec:targ}. We demonstrate this in Fig.~\ref{fig:exact_results}a, where the target sets are in green (solid for $\targetset$, dashed for $\targetset'$). When warm-starting from the original BRT $\vset_{\targetfunc}^*$ (cyan), we are able to recover the new BRT  $\vset_{\targetfunc'}^*$ (red) exactly, resulting in $\vset_{\warmfunc}^*$ (black). 
We similarly show the reverse case for a decreasing target set in Fig.~\ref{fig:exact_results}b. See Remark~\ref{rem:1} for details.

\subsection{Changing Control Authority}
In many applications the control authority can change over time. This can happen because of several reasons; for example, increasing the mass of a quadrotor leads to a reduction in its effective control authority.  We can explicitly modify $\controlspace$ when there is a change in the control bounds or in a model parameter which updates the effective control authority. When the control space is decreased, i.e. $\controlspace' \subseteq \controlspace$, initializing with the previously converged  value function will lead to $[\warmfunc(\state) = \vconv_{\targetfunc}(\state)]\geq\vconv_{\targetfunc'}(\state)$ 
and therefore exact convergence is guaranteed. Refer to Proposition~\ref{prop:controlbounds} and its proof in Appendix~\ref{subsec:cont}. 

To demonstrate this case of reduced control authority we vary the parameter $b$ in the system model \eqref{eqn:doubleint}.  When $b$ decreases, the effective control authority decreases.  In Fig.~\ref{fig:exact_results}c we compute the value function for $b=1$ (cyan).  We then compute the value function for $b = .8$ (red).  Finally, we warm-start from the original cyan value function and reach the new red value function exactly, as shown in black. We similarly show the reverse case for an increasing control authority in Fig.~\ref{fig:exact_results}d. See Remark~\ref{rem:2} for details.

\begin{table}[]\centering
\begin{tabular}{llll}
\multicolumn{4}{c}{Table 1: Runtime Analysis for Reachability Methods}                                                                                                                                                                                                                                                                                                                      \\ \hline
\multicolumn{1}{|l|}{}                                                                                 & \multicolumn{1}{c|}{Standard}                                                    & \multicolumn{1}{c|}{Warm-Start}                                                  & \multicolumn{1}{c|}{Discounted}                                                                              \\ \hline
\multicolumn{1}{|l|}{a) Increasing $\targetset$ (exact)}                                                          & \multicolumn{1}{l|}{\begin{tabular}[c]{@{}l@{}}6.4s, \\ 115 steps\end{tabular}}  & \multicolumn{1}{l|}{\begin{tabular}[c]{@{}l@{}}6.0s, \\ 109 steps\end{tabular}}  & \multicolumn{1}{l|}{\begin{tabular}[c]{@{}l@{}}12.5s, \\ 231 steps\end{tabular}}                             \\ \hline
\multicolumn{1}{|l|}{b) Decreasing $\targetset$ (conserv)}                                                        & \multicolumn{1}{l|}{\begin{tabular}[c]{@{}l@{}}6.0s \\ 110 steps\end{tabular}}   & \multicolumn{1}{l|}{\begin{tabular}[c]{@{}l@{}}9.3s, \\ 169 steps\end{tabular}}  & \multicolumn{1}{l|}{\begin{tabular}[c]{@{}l@{}}21.2s, \\ 385 steps\end{tabular}}                             \\ \hline
\multicolumn{1}{|l|}{d) Decreasing $\controlspace$ (exact)}                                                          & \multicolumn{1}{l|}{\begin{tabular}[c]{@{}l@{}}6.5s, \\ 124 steps\end{tabular}}  & \multicolumn{1}{l|}{\begin{tabular}[c]{@{}l@{}}6.0s, \\ 111 steps\end{tabular}}  & \multicolumn{1}{l|}{\begin{tabular}[c]{@{}l@{}}20.3s, \\ 374 steps\end{tabular}}                             \\ \hline
\multicolumn{1}{|l|}{e) Increasing $\controlspace$ (conserv)}                                                        & \multicolumn{1}{l|}{\begin{tabular}[c]{@{}l@{}}6.8s, \\ 124 steps\end{tabular}}  & \multicolumn{1}{l|}{\begin{tabular}[c]{@{}l@{}}6.1s, \\ 111 steps\end{tabular}}  & \multicolumn{1}{l|}{\begin{tabular}[c]{@{}l@{}}20.5s, \\ 374 steps\end{tabular}}                             \\ \hline
\multicolumn{1}{|l|}{c) Increasing $\disturbancespace$ (exact)}                                                          & \multicolumn{1}{l|}{\begin{tabular}[c]{@{}l@{}}21.4s, \\ 311 steps\end{tabular}} & \multicolumn{1}{l|}{\begin{tabular}[c]{@{}l@{}}7.7s, \\ 112 steps\end{tabular}}  & \multicolumn{1}{l|}{\begin{tabular}[c]{@{}l@{}}13.3s, \\ 195 steps\end{tabular}}                             \\ \hline
\multicolumn{1}{|l|}{d) Decreasing $\disturbancespace$ (conserv)}                                                        & \multicolumn{1}{l|}{\begin{tabular}[c]{@{}l@{}}6.0s, \\ 110 steps\end{tabular}}  & \multicolumn{1}{l|}{\begin{tabular}[c]{@{}l@{}}11.7s, \\ 213 steps\end{tabular}} & \multicolumn{1}{l|}{\begin{tabular}[c]{@{}l@{}}19.0s, \\ 346 steps\end{tabular}}                             \\ \hline
\multicolumn{1}{|l|}{\begin{tabular}[c]{@{}l@{}}e) 10D quad\\ increasing m, $\disturbancespace$ (exact)\end{tabular}}    & \multicolumn{1}{l|}{\begin{tabular}[c]{@{}l@{}}3.7hr, \\ 86 steps\end{tabular}}  & \multicolumn{1}{l|}{\begin{tabular}[c]{@{}l@{}}2.8hr, \\ 65 steps\end{tabular}}  & \multicolumn{1}{l|}{\begin{tabular}[c]{@{}l@{}}\textgreater{}18 hr, \\ \textgreater{}401 steps\end{tabular}} \\ \hline
\multicolumn{1}{|l|}{\begin{tabular}[c]{@{}l@{}}f) 10D quad \\ decreasing m, $\disturbancespace$ (conserv)\end{tabular}} & \multicolumn{1}{l|}{\begin{tabular}[c]{@{}l@{}}.68hr,\\ 50 steps\end{tabular}}   & \multicolumn{1}{l|}{\begin{tabular}[c]{@{}l@{}}.67hr,\\ 48 steps\end{tabular}}   & \multicolumn{1}{l|}{\begin{tabular}[c]{@{}l@{}}1.12hr,\\ 82 steps\end{tabular}}                              \\ \hline
\end{tabular}\vspace{-2em}
\end{table}

\subsection{Changing Disturbance Authority}
Following similar logic to the previous example, we find that increasing $\disturbancespace$ to a larger $\disturbancespace'$ has the same effect on the value function as decreasing $\controlspace$ to $\controlspace'$. To demonstrate this, we change the disturbance bounds in our model \eqref{eqn:doubleint}.  In Fig.~\ref{fig:exact_results}e we compute the value function for $d\in[0,0]$, shown in blue.  We then compute the value function for $d \in [-4,4]$.  Finally, we warm-start from the original cyan value function and reach the new red value function exactly, as shown in black. We similarly show the reverse case for a decreasing disturbance authority in Fig.~\ref{fig:exact_results}f.

%% file: fancyexample.tex
\section{high-dimensional example}
The strength of warm-starting in reducing computation time is best seen in high-dimensional examples.  In this example we perform reachability analysis to provide safety guarantees for a 10D nonlinear near-hover quadcopter model from \cite{Bouffard12,Chen2016DecouplingJournal}. When the quadcopter experiences changes to its environment constraints or system dynamics (e.g. changes in mass or disturbances), it must update its safety guarantees appropriately.  

The 10D near-hover quadcopter dynamics has states $( p_x,  p_y,  p_z)$ denoting the position, $( v_x,  v_y,  v_z)$ for velocity, $(\theta_x, \theta_y)$ for pitch and roll, and $(\omega_x, \omega_y)$ for pitch and roll rates. Its controls are $(S_x, S_y)$, which respectively represent the desired pitch and roll angle, and $T_z$, which represents the vertical thrust. The disturbances are $(\dstb_x, \dstb_y, \dstb_z)$ which represents wind, and $g$ is gravity. Its model is:\vspace{-.5em}

\begin{equation}
\label{eq:Quad10D_dyn}
\begin{aligned}
\begin{array}{c}
	\left[
	\begin{array}{c}
	\dot p_x\\
	\dot v_x\\
	\dot\theta_x\\
	\dot\omega_x\\
	\dot p_y\\
	\dot v_y\\
	\dot\theta_y\\
	\dot\omega_y\\
	\dot p_z\\
	\dot v_z
	\end{array}
	\right]
	=
	\left[
	\begin{array}{c}
	 v_x + \dstb_x\\
	g \tan \theta_x\\
	-d_1 \theta_x + \omega_x\\
	-d_0 \theta_x + n_0 S_x\\
	 v_y + \dstb_y\\
	g \tan \theta_y\\
	-d_1 \theta_y + \omega_y\\
	-d_0 \theta_y + n_0 S_y\\
	 v_z + \dstb_z \\
	(k_T/m) T_z - g
	\end{array}
	\right].
\end{array}\\
\end{aligned}
\end{equation}

\noindent The parameters $d_0, d_1, n_0, k_T$, as well as the control bounds $\cset$ that we used were $d_0 = 10, d_1 = 8, n_0 = 10, k_T = 4.55, |\ctrl_x|, |\ctrl_y| \le 10 \text{ degrees}, 0 \le \ctrl_z \le 2g$.  As in \cite{Chen2016DecouplingJournal}, we can decompose this into two 4D systems and one 2D system.

In this example the initial mass is $m = 5$ and initial disturbances are $|\dstb_x|, |\dstb_y| \le 1, |\dstb_z| \le 1$.  As the quadcopter is flying, the mass increases to $m=5.25$ (say, due to rain accumulation or picking up a package), effectively decreasing the control bounds. In addition, disturbance bounds go up: $|\dstb_x|, |\dstb_y| \le 1.5$. In this scenario we can warm-start from the previously computed value function to update the safety guarantees exactly. The value function converges to the true solution (max error of 0.189 in the $p_x, p_y$ subsystems and 0.003 in the $p_z$ subsystem) in 66 steps (2.8 hours) instead of 87 steps (3.65 hours) for standard reachability. Discounted reachability still hadn't converged after 400 steps (18+ hours), with an error of 0.0034 in $p_x, p_y$ and .325 in the $p_z$ subsystem. 

If the mass and disturbances instead go down (say, to $m = 4.8, |\dstb_x|, |\dstb_y| \le .95$), we can guarantee that the warm-start solution will at best be exactly the new solution, and at worst will be a conservative solution.
As demonstrated in Sec.~\ref{sec:examples_conserv}, in practice we almost always converge to the correct solution, and this 10D example converges correctly as well (max error of $2.7e-05$ in the $p_x, p_y$ subsystems and .074 in the $p_z$ subsystem). Our warm-starting method took 48 steps, compared to 50 for standard reachability and 82 for discounting (with errors of $8.6e-06$ in the the $p_x, p_y$ subsystems and .004 in the $p_z$ subsystem).
Though warm-starting does not provide much computational benefit in this case, every iteration toward convergence provides a guaranteed safe over-approximation of the BRT, which is not true for standard reachability.\vspace{-.5em} 

%% file: conclusion.tex
\section{Discussion \& Conclusion}
Warm-starting infinite-horizon HJI reachability computations with intelligent initializations is beneficial because it may lead to a sizable reduction in computation time by reducing the number of iterations required for convergence.  In this paper we proved that warm-starting  
will provide guaranteed conservative safety analyses and controllers.  Moreover, when the initialization is under-conservative (i.e. $\warmfunc(\state) \geq \vfuncl^*(\state)$), we proved that the reachability analysis is guaranteed to converge to the true solution. 
We also showed several conditions for which exact convergence is achieved, and several cases that will either move closer to the correct safety guarantees or remain conservative with every iteration.
In practice we frequently converge to the correct solution regardless of the conservativeness of the initialization.

 We demonstrated these results through several examples, including a 10D quadcopter model experiencing changes in mass and disturbance bounds. We were able to accurately recover the updated value function representing the backwards reachable tube in fewer iterations than standard or discounted reachability. For our examples we find that warm-start reachability is $1.6$ times faster than standard reachablity and $6.2$ times faster than (untuned) discounting.

This new formulation opens the door to many different methods for solving HJI reachability problems efficiently. One direction would be to numerically parameterize the value function (for example, by different masses), then warm-start online using an interpolated intialization based on updated problem information.  Another exciting direction is to update the value function locally for local changes in the environment, or to use sparse or adaptive gridding of the state space for fast initializations.  Finally, we could use the conclusions drawn from this paper to inform a more tractable formulation of discounted reachability. \vspace{0em}

%% file: appendix.tex
\section{Appendix}

\subsection{Changing the Target Set}
\label{subsec:targ}
For a target function, $\targetfunc'(\state)$, we define $\costfunctional_{\targetfunc'}$ and $\vfunc_{\targetfunc'}(\state, \tvar)$ as~\eqref{eq:costfunctional} and~\eqref{eq:vl} by replacing $\targetfunc(\state)$ by $\targetfunc'(\state)$. 
We define $\vfunck(\state, \tvar)$ as the value function for $\targetset'$ when we warm-start from $\vconv_{\targetfunc}(\state)$.

%
\begin{proposition}
\label{prop:targetset}
    If  $\targetset \subseteq \targetset'$ (and hence $\targetfunc(\state) \leq \targetfunc'(\state)$) and we warm-start the value function computation for $\targetset'$ with $\warmfunc(\state) = \vconv_{\targetfunc}(\state)$, i.e., $\vfunc_{\warmfunc}(\state, 0) = \vconv_{\targetfunc}(\state)$, then 
    \begin{equation}
        \lim_{\tvar \rightarrow -\infty} \vfunc_{\warmfunc}(\state, \tvar) = \vconv_{\targetfunc'}(\state)
    \end{equation}
\end{proposition}
\begin{proof}
To prove this, it suffices to prove that 
$
    \forall \state \; \vconv_{\targetfunc'}(\state) \leq \vconv_{\targetfunc}
$. We have $\forall \state, \tvar < 0$, \vspace{-.5em}
\begin{equation}
    \begin{split}
        \vfunc_{\targetfunc'}(\state, \tvar) &= \max_{\cfunc}\min_{\dfunc}
     \inf_{\tdummy \in [\tvar,0]} \targetfunc'(\trajstandard(\tdummy)) \\
     & \leq \max_{\cfunc}\min_{\dfunc}
     \inf_{\tdummy \in [\tvar,0]} \targetfunc(\trajstandard(\tdummy)) 
     = \vfunc_{\targetfunc}(\state, \tvar)
    \end{split}
\end{equation}
As $\tvar \rightarrow -\infty$, $\vconv_{\targetfunc'}(\state) \leq \vconv_{\targetfunc}(\state)$. From Theorem~\ref{theorem:equality}, if we warm-start, $\vfunc_{\targetfunc'}(\state, 0) =\warmfunc(\state)$ where $\warmfunc(\state) \geq \vconv_{\targetfunc'}(\state)$, then $\vconv_{\warmfunc}(\state) = \vconv_{\targetfunc}$. Since,  $\warmfunc(\state) = \vconv_{\targetfunc}(\state) \geq \vconv_{\targetfunc'}(\state)$, Theorem~\ref{theorem:equality} holds. \vspace{0em}
\end{proof}
\begin{remark} \label{rem:1}
By reversing the proof with conditions $\targetset \supseteq \targetset'$, then $\warmfunc(\state) \leq \vfunck(\state,\tvar) \leq \vfuncl^*(\state) \; \forall \state,\tvar$.
\end{remark}\vspace{-.4em}

\subsection{Changing the Control Authority}
\label{subsec:cont}
For a control domain, $\controlspace'$, we define  $\vfunc_{\targetfunc'}(\state, \tvar)$ similar to~\eqref{eq:vl} by replacing $\ctrl \in \controlspace$ by $\ctrl \in \controlspace'$. 
We define $\vfunck(\state, \tvar)$ as the value function for $\controlspace'$ when we warm-start from $\vconv_{\targetfunc}(\state)$.


\begin{proposition}
\label{prop:controlbounds}
    If the effective control authority $\controlspace'$ reduces, i.e., $\controlspace' \subseteq \controlspace$ and if $\vfunc_{\warmfunc}(\state, 0) = \warmfunc(\state) =  \vconv_{\targetfunc}(\state)$, then
    \begin{equation}
    \lim_{\tvar  \rightarrow -\infty} \vfunc_{\warmfunc}(\state, 0) = \vconv_{\targetfunc'}(\state)
\end{equation}
\end{proposition}
\begin{proof}
To prove this, we need only prove that, 
$
    \forall \state \; \vconv_{\targetfunc'}(\state) \leq \vconv_{\targetfunc}(\state)
$. 
We have, $\forall \state, \tvar <0$,\vspace{-1em}
\begin{equation}
    \begin{split}
        \vfunc_{\targetfunc'}(\state, \tvar) &= \max_{\ctrl \in \controlspace'} \min_{\dstb \in \disturbancespace} \min \Big\{
     \inf_{\tdummy \in [\tvar,0)} \targetfunc(\trajstandard(\tdummy)),\targetfunc(\trajstandard(0))\Big\} \\
     & \leq \max_{\ctrl \in \controlspace} \min_{\dstb \in \disturbancespace} \min \Big\{
     \inf_{\tdummy \in [\tvar,0)} \targetfunc(\trajstandard(\tdummy)),\targetfunc(\trajstandard(0))\Big\}\\
     &= \vfunc_{\targetfunc}(\state, \tvar)
    \end{split}
\end{equation}
As $\tvar \rightarrow -\infty$, $\vconv_{\targetfunc'}(\state) \leq \vconv_{\targetfunc}(\state)$. 
From Theorem~\ref{theorem:equality}, if we warm-start with $\vfunc_{\warmfunc}(\state, 0) = \warmfunc(\state)$ with $\warmfunc(\state) \geq \vconv_{\targetfunc}(\state)$, then $\vconv_{\warmfunc}(\state) = \vconv_{\targetfunc'}(\state)$. Hence, $\warmfunc(\state) = \vconv_{(\targetfunc}(\state)$ satisfies Theorem~\ref{theorem:equality}.
\end{proof}
\begin{remark} \label{rem:2}
By reversing the proof with conditions $\controlspace \subseteq \controlspace'$, then $\warmfunc(\state) \leq \vfunck(\state,\tvar) \leq \vfuncl^*(\state) \; \forall \state,\tvar$.
\end{remark}\vspace{0em}